\documentclass[dvipsnames,letterpaper,11pt]{article}

\usepackage[utf8]{inputenc}
\usepackage[english]{babel}

\usepackage{booktabs}
\usepackage{multirow}

\usepackage{xspace}
\usepackage{bm}
\usepackage[tbtags]{amsmath}
\usepackage{amsfonts,amsthm,amssymb}
\usepackage{mathrsfs} 
\usepackage{mathtools}
\usepackage{thmtools}
\usepackage[]{mdframed}

\usepackage{tikz}

\usetikzlibrary{shadows.blur}

\usetikzlibrary{shapes.geometric}

\usepackage{verbatim}
\usepackage{subcaption}

\usepackage[ruled,linesnumbered]{algorithm2e}

\usepackage[margin=1in]{geometry}
\usepackage[pagebackref=true,colorlinks,urlcolor=blue,linkcolor=purple,citecolor=blue,bookmarks,bookmarksopen,bookmarksnumbered]{hyperref}

\usepackage[capitalise,nameinlink]{cleveref}

\usepackage{tcolorbox}

\usepackage{todonotes} 
\usepackage{tablefootnote}

\usepackage{mleftright}

\usepackage{dsfont}
\usepackage[T1]{fontenc}

\usepackage{physics}

\newtheorem{theorem}{Theorem}[section]
\newtheorem{lemma}[theorem]{Lemma}

\newtheorem{definition}[theorem]{Definition}

\newtheorem{proposition}[theorem]{Proposition}

\newtheorem{remark}[theorem]{Remark}

\DeclareMathOperator*{\E}{\mathbb{E}}
\newcommand{\N}{\mathbb{N}}

\newcommand{\F}{\mathbb{F}}

\newif\ifshownotes
\shownotestrue

\ifshownotes
	\newcommand{\hstnote}[1]{\textcolor{red}{\textbf{[Shengtang: #1]}}}
\else
	\newcommand{\hstnote}[1]{}
\fi

\numberwithin{equation}{section}

\title{
Bounded Independence for $k$-Min-Wise Hashing:\\
Tight Bounds and Limitations of Structured Hashing
}
\author{
Xue Chen\thanks{\texttt{xuechen1989@ustc.edu.cn}. University of Science and Technology of China.}
\and
Shengtang Huang\thanks{\texttt{peanuttang1320061044@gmail.com}. Center on Frontiers of Computing Studies, Peking University.}
\and
Xin Li\thanks{\texttt{lixints@cs.jhu.edu}. Department of Computer Science, Johns Hopkins University.}
\and
Haoran Wang\thanks{\texttt{whr.hrwang@gmail.com}. School of Mathematical Sciences, Peking University.}
}
\date{}

\begin{document}

\maketitle

\begin{abstract}
Min-wise hashing and its $k$-min-wise extension are fundamental tools in sampling, sketching, similarity estimation, etc. A standard approach to constructing such families is bounded independence. For ordinary min-wise hashing, the required degree of independence is fully understood: $\Theta(\log 1/\delta)$-wise independence is both sufficient and necessary. For $k$-min-wise hashing, however, the best previous result only showed that $O(k\log\log1/\delta+\log1/\delta)$-wise independence suffices, with no matching lower bound.

We give a tight characterization of the amount of bounded independence required for $k$-min-wise hashing, proving that $\Theta(k+\log 1/\delta)$-wise independence is both sufficient and necessary. This improves the previous upper bound and provides a matching lower bound. Consequently, the standard construction of bounded-independent hash families has seed length $O\bigl((k+\log 1/\delta)\cdot\log(N/\delta)\bigr)$. In particular, for polynomially small $\delta$ and any $\Omega(\log N) \le k \le N^{1 - c}$, it achieves the optimal seed length $O(k\log N)$.

We further investigate two standard low-independence hash families. For random affine functions over $\F_2$, which form a pairwise independent family, we show that the multiplicative error is $\Omega(\log N)$ even for ordinary min-wise hashing. For simple tabulation hashing, which is $3$-wise independent and performs well for ordinary min-wise hashing, we show that it incurs a multiplicative error $\Omega(N)$ for $k$-min-wise hashing whenever $k\ge 4$.
\end{abstract}

\thispagestyle{empty}

\newpage
\setcounter{page}{1}

\section{Introduction}\label{sec:introduction}

Min-wise hashing, introduced by Broder, Charikar, Frieze, and Mitzenmacher~\cite{BCFM98}, is a fundamental primitive in randomized algorithms, with applications to similarity and rarity estimation, near-duplicate detection, and $\ell_0$-sampling in streaming algorithms~\cite{Cohen_similarity_est,DM_similarity_rarity,Henzinger06,sampling_survey}. Intuitively, min-wise hashing uses a hash function to approximately simulate a uniformly random ordering: for any designated element, the probability that it receives the smallest hash value should be close to that under a random permutation. Its $k$-min-wise extension, introduced by Feigenblat, Porat, and Shiftan~\cite{minwise_FPS11}, requires an analogous guarantee for any prescribed set of at most $k$ elements: the probability that all of them receive smaller hash values than all remaining elements should be close to the corresponding probability under a uniformly random permutation.

\begin{definition}[$k$-Min-Wise Hashing]\label{def:k_min_wise}
We use $a = b \pm \delta$ to denote $a \in [b - \delta, b + \delta]$. Let $0<\delta<1$, and let $\mathcal H$ be a family of functions from $[N]$ to an ordered range $[M]$. We say that $\mathcal H$ is a $k$-min-wise hash family with multiplicative error $\delta$ if for every pair of sets $X$ and $Y\subseteq[N]$ satisfying $X\cap Y=\varnothing$ and $|Y|\le k$,
$$
    \Pr_{h\sim\mathcal H}\left[\max_{y\in Y}h(y)<\min_{x\in X}h(x)\right]
    =(1\pm\delta) \cdot \binom{|X| + |Y|}{|Y|}^{-1}.
$$
The case $k=1$ is called min-wise hashing.
\end{definition}

Under a fully random continuous ordering, the relative ordering of $X\cup Y$ is a uniformly random permutation. Hence every $|Y|$-subset is equally likely to occupy the first $|Y|$ positions, and the probability in the definition is exactly $\binom{|X|+|Y|}{|Y|}^{-1}$. For a finite range, ties between the elements of $X$ and $Y$ introduce a small additional error. Throughout the paper, we take $M=\Omega(N/\delta)$, for which a fully random function into $[M]$ already satisfies the definition within an $O(\delta)$ multiplicative error, as stated by Indyk~\cite{minwise_Indyk}.

A central goal of hashing is to construct such families using as little randomness as possible. One of the most natural approaches is bounded independence. A $t$-wise independent hash family agrees exactly with a fully random function on every collection of at most $t$ inputs and admits the standard polynomial construction with seed length essentially linear in $t$. This leads to the question of determining the amount of independence needed for $k$-min-wise hashing as a function of $k$ and $\delta$.

\subsection{Prior Works}

For ordinary min-wise hashing, Indyk~\cite{minwise_Indyk} proved that $O(\log1/\delta)$-wise independence suffices to achieve multiplicative error $\delta$. Pătraşcu and Thorup~\cite{PT_lb_t_wise_min_wise} proved a matching lower bound. Thus, for min-wise hashing, the power of bounded independence is completely characterized: the optimal degree is $\Theta(\log1/\delta)$.

For $k$-min-wise hashing, Feigenblat, Porat, and Shiftan~\cite{minwise_FPS11} proved that $O(k\log\log1/\delta + \log 1/\delta)$-wise independence suffices. No matching lower bound depending on $k$ was known. In particular, for $k=\Theta(\log N)$ and $\delta=N^{-c}$, their analysis combined with the standard construction of bounded independence (e.g. \cite{ABI86_construction_t_wise}) gives seed length $O(k\log N\log\log N)$.

There is also a line of work constructing min-wise hash families using general pseudorandomness tools. Saks, Srinivasan, Zhou, and Zuckerman~\cite{SSZZ00} related min-wise hashing to pseudorandom generators for combinatorial rectangles, and Gopalan and Yehudayoff~\cite{GY20} later improved the seed length. Building on pseudorandom generators for combinatorial rectangles together with additional pseudorandomness techniques, Chen, Huang, and Li~\cite{CHL26} recently constructed explicit $k$-min-wise families using $O(k\log N)$ random bits and achieving multiplicative error $2^{-O(\log N/\log\log N)}$ for $k=\log^{O(1)}N$. Their seed length is optimal up to constant factors, while the error is almost polynomially small. This raises the question of whether one can simultaneously obtain the optimal $O(k\log N)$ seed length and polynomially small error in the large-$k$ regime.

Besides general constructions, another line of work studies the min-wise behavior of concrete hash families. Two particularly classical examples are modular affine hashing and tabulation hashing, both belonging to the universal-hashing framework initiated by Carter and Wegman~\cite{CW_universal,WC_bounded_independence}. These families are simple to represent and evaluate and have been extensively studied in hash tables and randomized data structures. The one-dimensional modular affine family $h(x)=(ax+b)\bmod p$ is a canonical pairwise-independent construction. Nevertheless, Broder, Charikar, Frieze, and Mitzenmacher~\cite{BCFM98} and P\u{a}tra\c{s}cu and Thorup~\cite{PT_lb_t_wise_min_wise} showed that it incurs multiplicative error $\Omega(\log N)$ for ordinary min-wise hashing.

For simple tabulation hashing, the universe $[N]$ is viewed as $\Sigma^c$, where $|\Sigma|=N^{1/c}$, and a key $x=(x_1,\ldots,x_c)$ is hashed as $h(x)=\bigoplus_{i=1}^c T_i[x_i]$ using independent random tables. Although simple tabulation is only $3$-wise independent, P\u{a}tra\c{s}cu and Thorup~\cite{PT_tabulation} showed that it provides unexpectedly strong guarantees for several classical applications. In particular, for a set of $n$ keys, they proved a multiplicative min-wise error of $O\left(\frac{\log^2 N}{N^{1/c}}\right)$. P{\u{a}}tra{\c{s}}cu and Thorup~\cite{PT_twisted_tabulation} subsequently developed a cleaner related analysis for twisted tabulation and obtained the same error bound.  These results motivate asking whether the favorable min-wise behavior of structured low-independence families extends to the stronger $k$-min-wise property.

\subsection{Our Results}

We completely characterize the amount of bounded independence required for $k$-min-wise hashing.

\begin{theorem}[Informal Version of Results in \Cref{sec:bounded_independence}]
The degree of independence necessary and sufficient to guarantee $k$-min-wise hashing with multiplicative error $\delta$ is
$$
\Theta\left(k+\log\frac{1}{\delta}\right).
$$
\end{theorem}

The upper bound improves the result of Feigenblat, Porat, and Shiftan~\cite{minwise_FPS11} by removing the additional $\log\log 1/\delta$ factor multiplying $k$. We also provide a matching lower bound, proving that both the dependence on $k$ and the dependence on $\log 1/\delta$ are necessary. This gives a tight characterization of the power of bounded independence for constructing $k$-min-wise hash families.

The standard construction of bounded independence consequently has seed length $O\bigl((k+\log1/\delta)\cdot\log(N/\delta)\bigr)$. In particular, when $\delta=N^{-c}$ and $k=\Omega(\log N)$, this becomes $O(k\log N)$. In the polylogarithmic-$k$ regime, this seed length is optimal up to constant (see \Cref{sec:support-size-lower-bound}). Hence, for $\Omega(\log N)\le k\le\log^{O(1)}N$, we improve the almost-polynomial error of \cite{CHL26} to polynomially small error while preserving the optimal seed length. More generally, our bounded-independence characterization applies to every $k$ and continues to yield polynomially small error with seed length $O(k\log N)$.

Beyond the generic bounded-independence construction, we study two standard low-independence hash families and show that their additional structure does not lead to stronger $k$-min-wise guarantees.

We first consider random affine functions over $\F_2$, which are the natural higher-dimensional analogue of the classical modular affine family $h(x)=(ax+b)\bmod p$. The latter is known to incur multiplicative error $\Omega(\log N)$ for ordinary min-wise hashing~\cite{BCFM98,PT_lb_t_wise_min_wise}. We prove that the same limitation persists over binary vector spaces.

\begin{theorem}[Informal Version of \Cref{thm:min_wise_random_affine_not_good}]
There exists a set $S$ of size $N$ such that, for a uniformly random affine function $h(x)=Ax\oplus b$,
$$
\Pr\left[h(0)<\min_{s \in S} h(s)\right]=\Omega\left(\frac{\log N}{N}\right).
$$
Consequently, random affine hashing incurs multiplicative error $\Omega(\log N)$ for min-wise hashing.
\end{theorem}

We next consider simple tabulation hashing. P\u{a}tra\c{s}cu and Thorup~\cite{PT_tabulation} showed that, despite being only $3$-wise independent, simple tabulation provides strong guarantees for ordinary min-wise hashing, with error only $O\left(\frac{\log^2 N}{N^{1 / c}}\right)$. We show that this favorable behavior does not extend to $k$-min-wise hashing once $k\ge 4$.

\begin{theorem}[Informal Version of \Cref{thm:simple-tabulation}]
For every sufficiently large $N$, there exist disjoint sets $X$ and $Y$ with $|X|=N$ and $|Y|=4$ such that simple tabulation hashing satisfies
$$
\Pr\left[\max_{y \in Y} h(y)<\min_{x \in X} h(x)\right]=\Omega\left(\frac{1}{N^3}\right).
$$
Since the corresponding probability under a uniformly random ordering is $\Theta(1/N^4)$, simple tabulation incurs multiplicative error $\Omega(N)$ for $4$-min-wise hashing, and therefore for every $k\ge 4$.
\end{theorem}

\subsection{Relation to the Concurrent Preprints}

This paper merges and extends two concurrent and independent preprints: Chen, Huang, and Li, \emph{Constructions of $k$-Min-Wise Hash from Bounded Independence} (arXiv:2607.27157v1), and Wang, \emph{Limited Independence Suffices for Large-$k$ Min-wise Hashing} (arXiv:2607.10255v2). Both preprints independently established the $O(k+\log1/\delta)$ upper bound using different arguments. The present joint version gives a unified presentation and additionally includes the matching lower bound and the results on random affine hashing and simple tabulation hashing.

\subsection{Notation and Basic Definitions}

Let $[n]$ denote $\{1,2,\ldots,n\}$ and $\binom{n}{k}$ denote the binomial coefficient. Unless stated otherwise, all logarithms are to base $2$. For a finite family $\Omega$, we use $x\sim\Omega$ to denote that $x$ is sampled uniformly from $\Omega$. For three real variables $a$, $b$ and $\delta$, $a = b \pm \delta$ means the error between $a$ and $b$ is $|a - b| \le \delta$.

For a hash function $h$ and a set $S$, let
$$
h(S):=\big(h(s):s\in S\big)
$$
denote the \emph{multiset} of hash values of the elements of $S$. For every nonempty set $S$, we use the shorthand
$$
\min h(S):=\min_{s\in S}h(s),
\qquad
\max h(S):=\max_{s\in S}h(s).
$$

We recall the definition of bounded-independent hash families.

\begin{definition}[$t$-wise Independent Hash Family \cite{CW_universal,WC_bounded_independence}]\label{def:t_wise_independence}
    A hash family $\mathcal{H}$ from a domain $\mathcal{U}$ to a range $\mathcal{R}$ is called $t$-wise independent if, for any distinct inputs $x_1,\dots,x_t\in\mathcal{U}$ and any outputs $y_1,\dots,y_t\in\mathcal{R}$,
    $$
    \Pr_{h\sim\mathcal{H}}\big[h(x_1)=y_1,\dots,h(x_t)=y_t\big]
    =
    |\mathcal{R}|^{-t}.
    $$
\end{definition}

Explicit constructions of $t$-wise independent hash family require $\Theta\big(t\cdot(\log|\mathcal{U}|+\log|\mathcal{R}|)\big)$ random bits~\cite{ABI86_construction_t_wise}.
\section{Bounded Independence}\label{sec:bounded_independence}

We now determine the amount of bounded independence needed to guarantee $k$-min-wise hashing. For the upper bound, we fix disjoint sets $X$ and $Y$ and enumerate the possible values of $\theta:=\max h(Y)$. This expresses $\Pr[\max h(Y)<\min h(X)]$ as a weighted sum of the probabilities $\Pr[\min h(X)>\theta]$. We divide the sum according to the size of $\theta$. For small $\theta$, truncated inclusion--exclusion shows that $\Pr[\min h(X)>\theta]$ under bounded independence is close to the corresponding probability under fully random hashing. For large $\theta$, a tail bound for limited independence shows that the total contribution of this range is already small under both distributions. Combining the two estimates leaves only a small additive discrepancy, which is sufficiently small relative to the ideal probability $\binom{|X|+|Y|}{|Y|}^{-1}$ to yield the desired multiplicative error.

For the lower bound, we impose a parity constraint on the number of elements of $Y$ whose hash values fall in one half of the range. This preserves $(k-1)$-wise independence but biases the probability of the event $\max h(Y)<\min h(X)$, proving that $\Omega(k)$-wise independence is necessary. The additional $\Omega(\log1/\delta)$ lower bound follows from the min-wise lower bound of P\u{a}tra\c{s}cu and Thorup~\cite{PT_lb_t_wise_min_wise}. Together, these two bounds give the matching characterization.

\subsection{Upper Bound}\label{sec:upper_bound}

We first prove that $O(k+\log 1/\delta)$-wise independence suffices for $k$-min-wise hashing. This improves the previous upper bound
$O(k\log\log 1/\delta+\log 1/\delta)$ of~\cite{minwise_FPS11}.

\begin{theorem}\label{thm:k_log_1_eps_upper_bound}
    There exists a constant $c>0$ such that, if $M\ge N/\delta$, then every $c\cdot(k+\log 1/\delta)$-wise independent hash family from $[N]$ to $[M]$ is $k$-min-wise with error $3\delta$.
\end{theorem}

Let $X,Y\subseteq[N]$ be disjoint, let $|X|=s$, and, without loss of generality, consider the case $|Y|=k$. We use $\Pr_\ell$ to denote probability under an $\ell$-wise independent hash family, while $\Pr$ without a subscript denotes probability under fully random hashing.

We consider $(\ell+k)$-wise independence. Directly analyzing the joint event $\{\max h(Y)<\min h(X)\}$ is difficult because the hash values of the elements in $X$ and $Y$ may be correlated. We therefore enumerate $\theta:=\max h(Y)$ and obtain
$$
\Pr_{\ell+k}[\max h(Y)<\min h(X)]
=
\sum_{\theta=1}^{M}
\frac{\theta^k-(\theta-1)^k}{M^k}\cdot\Pr_\ell[A_\theta],
$$
$$
\Pr[\max h(Y)<\min h(X)]
=
\sum_{\theta=1}^{M}
\frac{\theta^k-(\theta-1)^k}{M^k}\cdot\Pr[A_\theta],
$$
where
$$
A_\theta=\{\min h(X)>\theta\}
\quad\text{and}\quad
Z_\theta=\big|h(X)\cap[1,\theta]\big|.
$$
Enumerating the hash values of the $k$ elements of $Y$ consumes $k$ degrees of independence, leaving $\ell$-wise independence on $X$.

Set $T=\frac{c\ell M}{s}\in\N$, where $c<1$ is a constant to be determined, and split the summation into two parts:
$$
\Pr_{\ell+k}[\max h(Y)<\min h(X)]
=
\sum_{\theta\le T-1}
\frac{\theta^k-(\theta-1)^k}{M^k}\Pr_\ell[A_\theta]
+
\sum_{\theta\ge T}
\frac{\theta^k-(\theta-1)^k}{M^k}\Pr_\ell[A_\theta].
$$

\paragraph{The First Part.}

We first consider $\Pr_\ell[A_\theta]$ and use the Bonferroni inequalities.

\begin{lemma}[Bonferroni Inequalities]\label{lem:Bonferroni_ineq}
    Let $A_1,\ldots,A_n$ be $n$ events. Then, for every $r=0,1,\ldots,n$,
    $$
    \Pr\left[\bigcap_{i=1}^{n}A_i\right]
    \begin{cases}
        \displaystyle
        \le
        \sum_{t=0}^{r}(-1)^t
        \sum_{\mathcal{I}\subseteq[n],\,|\mathcal{I}|=t}
        \Pr\left[\bigcap_{j\in\mathcal{I}}\overline{A_j}\right],
        & r\text{ is even},\\[5mm]
        \displaystyle
        \ge
        \sum_{t=0}^{r}(-1)^t
        \sum_{\mathcal{I}\subseteq[n],\,|\mathcal{I}|=t}
        \Pr\left[\bigcap_{j\in\mathcal{I}}\overline{A_j}\right],
        & r\text{ is odd}.
    \end{cases}
    $$
\end{lemma}

Assume that $\ell$ is even. Then
$$
\sum_{j=0}^{\ell-1}(-1)^j
\sum_{\substack{X'\subseteq X\\|X'|=j}}
\Pr_\ell[h(X')\subseteq[1,\theta]]
\le
\Pr_\ell[A_\theta]
\le
\sum_{j=0}^{\ell}(-1)^j
\sum_{\substack{X'\subseteq X\\|X'|=j}}
\Pr_\ell[h(X')\subseteq[1,\theta]].
$$
Consequently,
$$
\Pr_\ell[A_\theta]
=
\sum_{j=0}^{\ell-2}(-1)^j
\sum_{\substack{X'\subseteq X\\|X'|=j}}
\Pr_\ell[h(X')\subseteq[1,\theta]]
\pm\Delta,
$$
where
$$
\Delta
\le
\sum_{j=\ell-1}^{\ell}
\sum_{\substack{X'\subseteq X\\|X'|=j}}
\Pr_\ell[h(X')\subseteq[1,\theta]]\le
2\binom{s}{\ell-1}
\left(\frac{\theta}{M}\right)^{\ell-1}
\le
2\left(
    \frac{e s\theta}{M(\ell-1)}
\right)^{\ell-1}.
$$

The above analysis involves only subsets of size at most $\ell$. Hence the relevant probabilities are identical under fully random hashing and $\ell$-wise independent hashing, and the same argument gives
$$
\Pr[A_\theta]
=
\sum_{j=0}^{\ell-2}(-1)^j
\sum_{\substack{X'\subseteq X\\|X'|=j}}
\Pr_\ell[h(X')\subseteq[1,\theta]]
\pm\Delta.
$$
Combining the two estimates yields
$$
\Pr_\ell[A_\theta]
=
\Pr[A_\theta]\pm2\Delta
=
\Pr[A_\theta]
\pm
4\left(
    \frac{e s\theta}{M(\ell-1)}
\right)^{\ell-1}.
$$

Substituting this estimate into the summation, we obtain
$$
\begin{aligned}
    \sum_{\theta\le T-1}
    \frac{\theta^k-(\theta-1)^k}{M^k}
    \Pr_\ell[A_\theta]&=
    \sum_{\theta\le T-1}
    \frac{\theta^k-(\theta-1)^k}{M^k}
    \left(
        \Pr[A_\theta]
        \pm
        4\left(
            \frac{e s\theta}{M(\ell-1)}
        \right)^{\ell-1}
    \right)\\
    &=
    \sum_{\theta\le T-1}
    \frac{\theta^k-(\theta-1)^k}{M^k}
    \Pr[A_\theta]
    \pm
    \frac{4T^k}{M^k}
    \left(
        \frac{e sT}{M(\ell-1)}
    \right)^{\ell-1}.
\end{aligned}
$$

\paragraph{The Second Part.}

We again consider $\Pr_\ell[A_\theta]$, this time using a concentration inequality. We use the tail bound from \cite{Bellare_Rompel_94}.

\begin{lemma}[$t$-wise Independence Tail Bound~{\cite[Lemma~2.3]{Bellare_Rompel_94}}]\label{lem:tail_bound_t_wise_BR94}
    Let $t\ge4$ be even, and let $X_1,\ldots,X_n$ be $t$-wise independent random variables taking values in $[0,1]$. Let $X=X_1+\cdots+X_n$, let $\mu=\E[X]$, and let $A>0$. Then
    $$
    \Pr[|X-\mu|\ge A]
    \le
    8\left(
        \frac{t\mu+t^2}{A^2}
    \right)^{t/2}.
    $$
\end{lemma}

Let $4\le\ell'<\ell$ be an even parameter to be determined. Applying \Cref{lem:tail_bound_t_wise_BR94} with
$t=\ell'$ and $\mu=A=\E[Z_\theta]$, and assuming that
$\E[Z_\theta]\ge\ell'$ for every $\theta\ge T$, gives
$$
\text{both }\Pr_\ell[A_\theta]\text{ and }\Pr[A_\theta]
\le
8\left(
    \frac{2\ell'\E[Z_\theta]}{(\E[Z_\theta])^2}
\right)^{\ell'/2}
=
8\left(
    \frac{2\ell'M}{\theta s}
\right)^{\ell'/2}.
$$

Substituting this bound into the summation, we obtain
$$
\begin{aligned}
    \sum_{\theta\ge T}
    \frac{\theta^k-(\theta-1)^k}{M^k}
    \Pr_\ell[A_\theta]
    &\le
    \Pr_\ell[A_T]
    \sum_{\theta\in[T,2T]}
    \frac{\theta^k-(\theta-1)^k}{M^k}+
    \Pr_\ell[A_{2T}]
    \sum_{\theta\in[2T,4T]}
    \frac{\theta^k-(\theta-1)^k}{M^k}
    +\cdots\\
    &\le
    8\left(
        \frac{2\ell'M}{Ts}
    \right)^{\ell'/2}
    \left(
        \frac{2T}{M}
    \right)^k+
    8\left(
        \frac{2\ell'M}{2Ts}
    \right)^{\ell'/2}
    \left(
        \frac{4T}{M}
    \right)^k
    +\cdots\\
    &=
    8\left(
        \frac{2\ell'M}{Ts}
    \right)^{\ell'/2}
    \left(
        \frac{2T}{M}
    \right)^k
    \left(
        1+
        \left(\frac12\right)^{\ell'/2-k}
        +\cdots
    \right)\\
    &\le
    16\left(
        \frac{2\ell'M}{Ts}
    \right)^{\ell'/2}
    \left(
        \frac{2T}{M}
    \right)^k.
\end{aligned}
$$
The last inequality holds when $\ell'/2-k\ge1$, so that the geometric series is at most $2$.

The same argument applies under fully random hashing. Therefore,
$$
\begin{aligned}
    \sum_{\theta\ge T}
    \frac{\theta^k-(\theta-1)^k}{M^k}
    \Pr_\ell[A_\theta]
    &=
    \sum_{\theta\ge T}
    \frac{\theta^k-(\theta-1)^k}{M^k}
    \Pr[A_\theta]\pm
    32\left(
        \frac{2\ell'M}{Ts}
    \right)^{\ell'/2}
    \left(
        \frac{2T}{M}
    \right)^k.
\end{aligned}
$$

\paragraph{Combining the Estimates.}

Adding the two parts gives
\begin{align*}
    \sum_{\theta=1}^{M}
    \frac{\theta^k-(\theta-1)^k}{M^k}
    \Pr_\ell[A_\theta]
    &=
    \sum_{\theta=1}^{M}
    \frac{\theta^k-(\theta-1)^k}{M^k}
    \Pr[A_\theta]\\
    &\quad\pm
    \left(
        \frac{4T^k}{M^k}
        \left(
            \frac{e sT}{M(\ell-1)}
        \right)^{\ell-1}
        +
        32\left(
            \frac{2\ell'M}{Ts}
        \right)^{\ell'/2}
        \left(
            \frac{2T}{M}
        \right)^k
    \right).
\end{align*}

We now choose the parameters. Define an auxiliary error parameter
$\delta'=\delta/2^z$, where $z\in\N$ will be determined later.

Let $\ell'=\ell/30$, assuming that $60\mid\ell$ so that $\ell'$ is even. Set $c=1/6$, and hence
$$
T=\frac{\ell M}{6s} \in \N.
$$
This guarantees that
$$
\E[Z_\theta]=\frac{\theta s}{M}\ge\ell'
\qquad
\text{for every }\theta\ge T.
$$
Moreover, assuming $\ell\ge40$,
$$
\frac{e sT}{M(\ell-1)}
=
\frac{e\ell}{6(\ell-1)}
<
\frac12,
\qquad
\frac{2\ell'M}{Ts}
=
\frac25
<
\frac12.
$$

If we further assume that
$$
\frac{\ell'}{2}\ge\log\frac{32}{\delta'}
\qquad\text{and}\qquad
\ell-1\ge\log\frac{4}{\delta'},
$$
then
$$
\text{the total additive error}
\le
\frac{T^k}{M^k}\delta'
+
\frac{(2T)^k}{M^k}\delta'.
$$

We next choose $z$ and $\ell$ to satisfy all the constraints above:
\begin{equation}\label{eq:constraints}
    \begin{cases}
        60\mid\ell,\quad \ell\ge120;\\
        \ell\ge60k+60;\\
        \ell\ge60(z+\log 1/\delta+5).
    \end{cases}
\end{equation}
We also require
$$
\frac{T^k}{M^k}\delta'
\le
\frac{(2T)^k}{M^k}\cdot\frac{\delta}{2^z}
\le
\frac{(k/e)^k}{(s+k)^k}\delta
\le
\frac{\delta}{\binom{s+k}{k}}.
$$

It suffices to consider the case $s\ge k$, since the case $s<k$ is handled trivially by $2k$-wise independence. We therefore relax the preceding requirement to
$$
\frac{(2T)^k}{M^k}\cdot\frac{\delta}{2^z}
\le
\frac{(k/e)^k}{(2s)^k}\delta.
$$
Substituting $T=\ell M/(6s)$, this is equivalent to
$$
\left(\frac{2e}{3}\right)^k\ell^k
\le
k^k2^z.
$$
Taking logarithms gives
\begin{equation}\label{eq:z_ell_relation}
    z
    \ge
    k\log\frac{2e}{3}
    +
    k\log\frac{\ell}{k}.
\end{equation}

Let $L=\left\lceil\log1/\delta\right\rceil.$ We choose $z=15k+L+5$ and $\ell=60(z+L+5)$. Thus $\ell=900k+120L+600$, which clearly satisfies all the constraints in~\eqref{eq:constraints}.

To verify~\eqref{eq:z_ell_relation}, we use
$\log(1+x)\le1.45x$ for $x\ge0$:
\begin{align*}
    k\log\left(\frac{\ell}{k}\right)
    &=
    k\log\left(
        900\left(
            1+\frac{120L+600}{900k}
        \right)
    \right)\\
    &=
    k\log900
    +
    k\log\left(
        1+\frac{120L+600}{900k}
    \right)\\
    &\le
    9.82k
    +
    1.45k\left(
        \frac{120L+600}{900k}
    \right)\\
    &\le
    9.82k+0.194L+0.968.
\end{align*}
Adding the constant term
$k\log(2e/3)\le0.86k$, the right-hand side of~\eqref{eq:z_ell_relation} is at most
$$
0.86k+9.82k+0.194L+0.968
=
10.68k+0.194L+0.968.
$$
This is smaller than our choice $z=15k+L+5$, so~\eqref{eq:z_ell_relation} holds.

\paragraph{Conclusion.}

We have shown that the muliplicative error of
$\Pr_{\ell+k}[\max h(Y)<\min h(X)]$ relative to
$\Pr[\max h(Y)<\min h(X)]$ is at most $2\delta$.
It remains to show that a fully random hash function already satisfies the $k$-min-wise property with error $\delta$ when $M\ge N/\delta$, as observed by Indyk~\cite{minwise_Indyk}. The following is the formal proof.

\begin{lemma}\label{lem:k_min_wise_full_random}
    Let $X,Y\subseteq[N]$ be disjoint sets with $|X|=s$ and $|Y|=k$, and let $\sigma:[N]\to[M]$ be a fully random function. Then
    $$
    \frac{1}{\binom{s+k}{k}}
    \left(
        1-\frac{s+k}{M}
    \right)
    \le
    \Pr_\sigma[\max\sigma(Y)<\min\sigma(X)]
    \le
    \frac{1}{\binom{s+k}{k}}.
    $$
    In particular, since $s+k\le N$, the multiplicative error is at most $\delta$ whenever $M\ge N/\delta$.
\end{lemma}

\begin{proof}
    The exact probability
    $P=\Pr_\sigma[\max\sigma(Y)<\min\sigma(X)]$ is
    \begin{align*}
        P
        &=
        \sum_{\theta=1}^{M}
        \frac{\theta^k-(\theta-1)^k}{M^k}
        \left(
            1-\frac{\theta}{M}
        \right)^s\\
        &=
        \sum_{\theta=1}^{M}
        \int_{(\theta-1)/M}^{\theta/M}
        kx^{k-1}
        \left(
            1-\frac{\theta}{M}
        \right)^s
        \dd x.
    \end{align*}

    The ideal probability in the absence of collisions is given by the Beta integral
    $$
    I
    =
    \int_{0}^{1}
    kx^{k-1}(1-x)^s
    \dd x
    =
    \frac{1}{\binom{s+k}{k}}.
    $$

    For every
    $x\in\big[(\theta-1)/M,\theta/M\big]$, we have
    $1-\theta/M\le1-x$, and hence
    $$
    P
    \le
    \sum_{\theta=1}^{M}
    \int_{(\theta-1)/M}^{\theta/M}
    kx^{k-1}(1-x)^s
    \dd x
    =
    I.
    $$

    For the lower bound, we use
    $1-\theta/M\ge1-x-1/M$ and convexity:
    $$
    \left(
        1-\frac{\theta}{M}
    \right)^s
    \ge
    (1-x)^s
    -
    \frac{s}{M}(1-x)^{s-1}.
    $$
    Integrating over all intervals gives
    $$
    P
    \ge
    I
    -
    \frac{s}{M}
    \int_{0}^{1}
    kx^{k-1}(1-x)^{s-1}
    \dd x=
    I
    -
    \frac{s}{M}
    \frac{1}{\binom{s+k-1}{k}}
    =
    I\left(
        1-\frac{s+k}{M}
    \right).
    $$

    Therefore,
    $$
    I\left(
        1-\frac{s+k}{M}
    \right)
    \le
    P
    \le
    I.
    $$
    Thus the multiplicative error is at most $(s+k)/M$. Since $s+k\le N$, it is at most $\delta$ whenever $M\ge N/\delta$.
\end{proof}

Combining the two parts, the multiplicative error relative to
$\binom{s+k}{k}^{-1}$ is at most $3\delta$, while
$\ell+k=\Theta(k+\log 1/\delta)$. This proves
\Cref{thm:k_log_1_eps_upper_bound}.

\paragraph{An Alternative Argument via Subset Restriction.}
The bounded-independence argument above handles the second range by applying the tail inequality at a reduced moment order $\ell'<\ell$. We also observe that the intermediate range can be controlled in a different way, without reducing the moment order, by restricting attention to a suitably chosen subset of the competing set. Write $p_\theta := \frac{\theta}{M}$ and $\mu_\theta := sp_\theta = \E[Z_\theta]$.

For a subset $I \subseteq X$ of size $b$, let $A_\theta(I) := \{\min h(I) > \theta\}$. The Bonferroni calculation above applies equally to $I$ and gives
\begin{equation}\label{eq:subset-bonferroni}
  \left|
    \Pr_\ell[A_\theta(I)] - \Pr[A_\theta(I)]
  \right|
  \le
  4 \left( \frac{ebp_\theta}{\ell-1} \right)^{\ell-1}.
\end{equation}
Thus the error is controlled by the effective mean $bp_\theta$, rather than by the full size $s$.

For small thresholds, say $\mu_\theta \le \alpha(\ell-1)$ for a sufficiently small constant $\alpha > 0$, we simply take $I = X$. Then~\eqref{eq:subset-bonferroni} gives
$$
  \left|
    \Pr_\ell[A_\theta] - \Pr[A_\theta]
  \right|
  \le
  4(e\alpha)^{\ell-1}
  =
  e^{-\Omega(\ell)}.
$$

For the larger intermediate range $\alpha(\ell-1) < \mu_\theta \le R\ell$, where $R > 2$ is a fixed constant, choose a subset $X'_\theta \subseteq X$ of size $b_\theta := \left\lfloor \frac{\alpha(\ell-1)}{p_\theta} \right\rfloor$. Then $b_\theta p_\theta = \Theta(\ell)$, while $A_\theta \subseteq A_\theta(X'_\theta)$. Consequently,
$$
  \Pr_\ell[A_\theta]
  \le
  \Pr_\ell[A_\theta(X'_\theta)] 
  \le
  (1-p_\theta)^{b_\theta} + 4 \left( \frac{eb_\theta p_\theta}{\ell-1} \right)^{\ell-1} 
  =
  e^{-\Omega(\ell)}.
$$
Since also $\Pr[A_\theta] = (1-p_\theta)^s \le e^{-\mu_\theta} = e^{-\Omega(\ell)}$, we obtain
$$
  \left|
    \Pr_\ell[A_\theta] - \Pr[A_\theta]
  \right|
  =
  e^{-\Omega(\ell)}
$$
throughout this intermediate range.

For the remaining thresholds with $\mu_\theta > R\ell$, one may apply the tail inequality using the full moment order $\ell$:
$$
  \Pr_\ell[A_\theta],\, \Pr[A_\theta]
  \le
  8 \left( \frac{2\ell}{\mu_\theta} \right)^{\ell/2}.
$$
The same dyadic decomposition as above then gives a convergent geometric series, since $\ell/2 > k$.

Together with the same weighted summation as in the main argument, this yields an alternative proof: Below the scale $R\ell$, one either applies Bonferroni directly to $X$ or to a suitably chosen subset $X'_\theta$; above that scale, the full-order tail bound applies.

\subsection{Lower Bound}\label{sec:lower_bound}

We now prove that any degree of independence that universally guarantees $k$-min-wise hashing with multiplicative error $\delta$ must be at least $\Omega(k+\log1/\delta)$, matching the upper bound showed before.

We first prove the $\Omega(k)$ lower bound. Fix arbitrary disjoint sets $X,Y$ with $|X|=s$ and $|Y|=k$, where $s$ is any positive integer. For convenience, let the hash values lie in $[0,1)$. The elements of $X$ are hashed independently and uniformly. For $Y$, first sample the hash values of $k-1$ elements, and then sample the last one so that the number of elements of $Y$ landing in the upper half $[1/2,1)$ is even. Once the half containing an element has been determined, its position within that half is uniform. The hash values on $X$ and $Y$ are mutually independent.

This distribution is exactly $(k-1)$-wise independent: the parity constraint affects only the joint distribution of all $k$ elements, while every collection of at most $k-1$ elements remains independent and uniform. This method of imposing a parity constraint on the numbers of elements landing in the two halves was also used by~\cite{PT_lb_t_wise_min_wise} in their lower bound for min-wise hashing.

Consider the event
$$
E=\{\max h(Y)<\min h(X)\},
$$
and let $\mathcal{P}$ denote the event that an even number of elements of $Y$ land in the upper half. The distribution above is precisely the fully random distribution conditioned on $\mathcal{P}$, and
$\Pr[\mathcal{P}]=1/2$. Under fully random hashing,
$$
\Pr[E]
=
\binom{s+k}{k}^{-1}.
$$

Let $N=s+k$, and let $Z$ be the number of elements of $X\cup Y$ whose hash values lie in the lower half $[0,1/2)$. A fully random hash function is equivalently generated in two steps: first sample $N$ independent uniform values and sort them; then assign the $N$ elements of $X\cup Y$ to the resulting positions according to a uniformly random permutation. Since the hash values are independent and identically distributed, this random permutation is independent of the sorted values.

The event $E$ only requires that the first $k$ positions be occupied by the elements of $Y$, whereas $Z$ depends only on how many of the sorted values are below $1/2$. Conditioning on $E$ therefore does not change the distribution of $Z$, and
$$
Z\mid E
\sim
\operatorname{Bin}\left(
    N,\frac12
\right),
$$
where $\operatorname{Bin}(\cdot, \cdot)$ denotes the binomial distribution.

Conditioned on $E$, if $Z=j<k$, then exactly $k-j$ elements of $Y$ lie in the upper half; if $Z\ge k$, then all elements of $Y$ lie in the lower half. It follows that
\begin{align*}
    \Pr[\mathcal{P}\mid E]
    &=
    \sum_{j = k - 2, k - 4, \ldots} \Pr[Z = j] + \sum_{j = k}^N \Pr[Z = j]\\
    &=
    \frac{1}{2^N}
    \left(
        \sum_{j=k-2,k-4,\ldots}
        \binom{N}{j}
        +
        \sum_{j=k}^{N}
        \binom{N}{j}
    \right)\\
    &=
    \frac12\left(
        1+
        \frac{1}{2^{N-1}}
        \sum_{j=k}^{N-1}
        \binom{N-1}{j}
    \right)\\
    &=
    \frac12\left(
        1+
        \Pr\left[
            \operatorname{Bin}\left(
                N-1,\frac12
            \right)
            \ge k
        \right]
    \right).
\end{align*}

Since the $(k-1)$-wise independent distribution above is the fully random distribution conditioned on $\mathcal{P}$, Bayes' rule gives
$$
    \Pr[E\mid\mathcal{P}]
    =
    \frac{
        \Pr[E]\Pr[\mathcal{P}\mid E]
    }{
        \Pr[\mathcal{P}]
    }=
    \left(
        1+
        \Pr\left[
            \operatorname{Bin}\left(
                s+k-1,\frac12
            \right)
            \ge k
        \right]
    \right)
    \binom{s+k}{k}^{-1}.
$$

Thus the multiplicative error of this $(k-1)$-wise independent distribution is exactly
$$
\varepsilon_{k,s}
=
\Pr\left[
    \operatorname{Bin}\left(
        s+k-1,\frac12
    \right)
    \ge k
\right].
$$
This formula holds for every $s$. If $s\ge k$, then
$\varepsilon_{k,s}\ge1/2$, while
$\varepsilon_{k,s}\to1$ as $s\to\infty$. Consequently, for every $\delta<1$, one can choose $X$ sufficiently large that
$\varepsilon_{k,s}>\delta$. Hence $(k-1)$-wise independence does not suffice to guarantee the $k$-min-wise property with error $\delta$, proving an $\Omega(k)$ lower bound.

The construction can also be discretized over a sufficiently large even range $[M]$: first determine whether each element lies in the lower or upper half of the range, and then sample it uniformly within the selected half. As $M$ grows, the probability of the strict-ordering event approaches that in the continuous construction, so the lower bound continues to hold.

On the other hand, for the $\Omega(\log 1 / \delta)$ lower bound, we notice that the definition of $k$-min-wise hashing includes ordinary min-wise hashing as the case $|Y|=1$. Hence we can use the following lower bound of Pătraşcu and Thorup \cite{PT_lb_t_wise_min_wise} for ordinary min-wise hashing.

\begin{theorem}[{\cite[Theorem~3.1]{PT_lb_t_wise_min_wise}}]\label{thm:PT_lb}
    For every set $X$ of $N$ keys and another key $y$, there exists a $t$-wise independent hash family
    $$
    \mathcal{H}
    =
    \{h:X\cup\{y\}\to[0,1)\}
    $$
    such that
    $$
    \Pr_{h\sim\mathcal{H}}
    [h(y)<\min h(X)]
    =
    \frac{1+2^{-O(t)}}{N+1}.
    $$
\end{theorem}

This shows that achieving multiplicative error $\delta$ requires $t=\Omega(\log1/\delta)$. Combining this with the $\Omega(k)$ lower bound gives
$$
t
=
\Omega\left(
    \max\left\{
        k,\log\frac{1}{\delta}
    \right\}
\right)
=
\Omega\left(
    k+\log\frac{1}{\delta}
\right).
$$
Equivalently, when $\delta<2^{-k}$, the term $\log 1/\delta$ dominates, whereas when $\delta\ge2^{-k}$, the term $k$ dominates.
\section{Random Affine Functions}

\newcommand{\msb}{\mathsf{msb}}

We next examine a concrete bounded-independent family. For the classical one-dimensional modular affine family
$$
h(x)=(ax+b)\bmod p,
$$
\cite{BCFM98} and \cite{PT_lb_t_wise_min_wise} showed that the multiplicative error for min-wise hashing can be $\Omega(\log N)$. We consider its natural higher-dimensional analogue over $\F_2$,
$$
h(x)=Ax\oplus b,
$$
where $A$ and $b$ are chosen uniformly. Although this family is pairwise independent, we show that it still incurs multiplicative error $\Omega(\log N)$ for min-wise hashing. Consequently, it also fails to satisfy the stronger $k$-min-wise guarantee.

Throughout this section, vectors in $\F_2^\ell$ are ordered lexicographically, with the first coordinate being the most significant bit. For every nonzero $y\in\F_2^\ell$, let $\msb(y)$ denote the index of its first nonzero coordinate.

\begin{theorem}\label{thm:min_wise_random_affine_not_good}
Let $h(x)=Ax\oplus b$ be a random affine function from $\F_2^u$ to $\F_2^\ell$, where $\ell=2u$, and where the matrix $A\in\F_2^{\ell\times u}$ and the vector $b\in\F_2^\ell$ are chosen uniformly and independently. For every integer $1\le m\le u-1$, there exists a set $S\subseteq\F_2^u\setminus\{0\}$ of size $N=2^m$ such that
$$
\Pr_{A,b}[h(0)<\min h(S)]
\ge
(1-2^{-u})\frac{\log N+1}{4N-2}
=
\Theta\left(\frac{\log N}{N}\right).
$$
\end{theorem}

\begin{proof}
Let $V\subseteq\F_2^u$ be an $(m+1)$-dimensional subspace, and fix a nonzero linear functional $\varphi\in V^*$, where $V^*$ denotes the dual space of $V$. Define
$$
S:=\{x\in V:\varphi(x)=1\}.
$$
Since $\varphi$ is nonzero, $S$ is an affine hyperplane in $V$. Hence $|S|=2^m=N$, and $0\notin S$.

Let $\mathcal E$ denote the event that the restriction $A|_V$ is injective. For every nonzero vector $x\in V$, we have $\Pr[Ax=0]=2^{-\ell}=2^{-2u}$. Therefore, by the union bound and the assumption $m+1\le u$,
$$
\Pr[\mathcal E]
\ge
1-(2^{m+1}-1)2^{-2u}
\ge
1-2^{-u}.
$$

We now analyze the distribution conditioned on the event $\mathcal E$. Let $U:=A(V)$, and write $L:=A|_V:V\to U$. Then $L$ is a linear isomorphism, while $\mathcal Y:=A(S)$ is an affine hyperplane in $U$ that does not contain $0$.

Furthermore, conditioned on $\mathcal E$ and on the image space $U$, the map $L$ is uniformly distributed over all linear isomorphisms from $V$ to $U$. Therefore, the nonzero linear functional
$$
\psi:=\varphi\circ L^{-1}\in U^*
$$
is uniformly distributed over $U^*\setminus\{0\}$, and $\mathcal Y=\{y\in U:\psi(y)=1\}$.

Now fix $A$ and consider the random shift $b$. For every nonzero vector $y\in\F_2^\ell$, the vectors $b$ and $b\oplus y$ first differ at coordinate $\msb(y)$, and hence
$$
b<b\oplus y
\quad\Longleftrightarrow\quad
b_{\msb(y)}=0.
$$
Let $\msb(\mathcal Y):=\{\msb(y):y\in\mathcal Y\}$. Since the coordinates of $b$ are mutually independent and uniform,
$$
\Pr_b[h(0)<\min h(S)\mid A]
=
\Pr_b[b<b\oplus y,\ \forall y\in\mathcal Y\mid A]
=
\Pr_b[b_{\msb(y)}=0,\ \forall y\in\mathcal Y\mid A]
=
2^{-|\msb(\mathcal Y)|}.
$$

Let $e_1,\ldots,e_{m+1}$ be the reduced echelon basis of $U$, uniquely determined by $U$, satisfying
$$
\msb(e_1)<\cdots<\msb(e_{m+1}).
$$
For every $y=\sum_{i=1}^{m+1}\alpha_i e_i$, if $t$ is the smallest index such that $\alpha_t=1$, then $\msb(y)=\msb(e_t)$.

With respect to this basis, write $\psi$ as
$$
\psi\left(\sum_{i=1}^{m+1}\alpha_i e_i\right)
=
\sum_{i=1}^{m+1}c_i\alpha_i.
$$
Since $\psi$ is uniformly distributed over $U^*\setminus\{0\}$, the coefficient vector $c=(c_1,\ldots,c_{m+1})$ is uniformly distributed over $\F_2^{m+1}\setminus\{0\}$. Moreover,
$$
\mathcal Y
=
\left\{
\sum_{i=1}^{m+1}\alpha_i e_i:
\sum_{i=1}^{m+1}c_i\alpha_i=1
\right\}.
$$

Fix $t\in[m+1]$. The index $\msb(e_t)$ belongs to $\msb(\mathcal Y)$ if and only if there exists a coordinate vector satisfying $\alpha_1=\cdots=\alpha_{t-1}=0$ and $\alpha_t=1$ such that
$$
c_t+\sum_{i=t+1}^{m+1}c_i\alpha_i=1.
$$
This equation is solvable if and only if $c_t,\ldots,c_{m+1}$ are not all zero.

Let $T:=\max\{i:c_i=1\}$. The preceding condition is equivalent to $t\le T$, and hence $|\msb(\mathcal Y)|=T$. Since $c$ is uniformly distributed over all nonzero vectors,
$$
\Pr[T=j\mid\mathcal E]
=
\frac{2^{j-1}}{2^{m+1}-1},
\qquad j\in[m+1].
$$
Therefore,
$$
\E[2^{-|\msb(\mathcal Y)|}\mid\mathcal E]
=
\sum_{j=1}^{m+1}
\frac{2^{j-1}}{2^{m+1}-1}2^{-j}
=
\frac{m+1}{2(2^{m+1}-1)}.
$$
Finally,
$$
\Pr[h(0)<\min h(S)]
\ge
\Pr[\mathcal E]\cdot
\E[2^{-|\msb(\mathcal Y)|}\mid\mathcal E]
\ge
(1-2^{-u})\frac{m+1}{2(2^{m+1}-1)}
=
(1-2^{-u})\frac{\log N+1}{4N-2}.
$$
\end{proof}
\section{Simple Tabulation Hashing}

Simple tabulation hashing is one of the standard constructions used for pairwise-independent hashing; in fact, it is $3$-wise independent. P\u{a}tra\c{s}cu and Thorup~\cite{PT_tabulation} showed that, despite its limited independence, simple tabulation provides strong guarantees for ordinary min-wise hashing. We show that this positive behavior does not extend to $k$-min-wise hashing. More precisely, for every $k\ge 4$, simple tabulation incurs multiplicative error $\Omega(N)$. Our result can therefore be viewed as a negative extension of their analysis from ordinary min-wise hashing to its $k$-min-wise generalization.

Recall that a simple tabulation hash function $h:\Sigma^c\to\{0,1\}^w$ is defined by
$$
h(x_1,\ldots,x_c)=T_1[x_1]\oplus\cdots\oplus T_c[x_c],
$$
where all table entries $T_i[a]$ are mutually independent and uniformly distributed in $\{0,1\}^w$. We identify a hash value with the corresponding $w$-bit binary fraction in $[0,1)$. A pair $\alpha=(i,a)\in[c]\times\Sigma$ is called a \emph{position-character}, and we write $T[\alpha]:=T_i[a]$. For a key $x=(x_1,\ldots,x_c)$, let
$$
P(x):=\{(i,x_i):i\in[c]\}
$$
denote the set of its position-characters.

\begin{theorem}\label{thm:simple-tabulation}
Let $c\ge 2$ and let the hash range have size $M=2^w$. For every positive integer $N$ satisfying
$$
M\ge 2N
\qquad\text{and}\qquad
(|\Sigma|-2)|\Sigma|^{c-1}\ge N,
$$
there exist disjoint sets $X,Y\subseteq\Sigma^c$ with $|X|=N$ and $|Y|=4$ such that
$$
\Pr\left[\max h(Y)<\min h(X)\right]\ge\frac{1}{128N^3}.
$$
Consequently,
$$
\frac{\Pr\left[\max h(Y)<\min h(X)\right]}{\binom{N+4}{4}^{-1}}>\frac{N}{3072}.
$$
In particular, simple tabulation hashing incurs multiplicative error $\Omega(N)$ for $4$-min-wise hashing.
\end{theorem}

The proof consists of two lemmas. The first exploits the xor dependency among four keys forming a combinatorial rectangle. The second uses an ordered exposure of the position-characters to control the hash values of the remaining keys.

\begin{lemma}\label{lem:tabulation-rectangle}
Fix two distinct characters $0,1\in\Sigma$, and define
$$
y_{ab}:=(a,b,0,\ldots,0),\qquad a,b\in\{0,1\}.
$$
Let $Y:=\{y_{00},y_{01},y_{10},y_{11}\}$. For every dyadic threshold $p=2^{-r}$ with $r\le w$,
$$
\Pr[\max h(Y)<p]=p^3.
$$
\end{lemma}

\begin{proof}
Write $H_{ab}:=h(y_{ab})$. Every position-character appearing among the four keys occurs an even number of times. Therefore,
$$
H_{00}\oplus H_{01}\oplus H_{10}\oplus H_{11}=0,
$$
and hence
$$
H_{11}=H_{00}\oplus H_{01}\oplus H_{10}.
$$
Moreover, $H_{00},H_{01},H_{10}$ are independent and uniformly distributed, since simple tabulation is $3$-wise independent.

The event $H_{ab}<p$ means precisely that the first $r$ bits of $H_{ab}$ are all zero. Thus, if $H_{00},H_{01},H_{10}<p$, then the first $r$ bits of their xor are also zero, and hence $H_{11}<p$. Consequently,
$$
\{\max h(Y)<p\}=\{H_{00}<p,\ H_{01}<p,\ H_{10}<p\}.
$$
Since $H_{00},H_{01},H_{10}$ are independent and uniform,
$$
\Pr[\max h(Y)<p]=p^3.
$$
\end{proof}

\begin{lemma}\label{lem:tabulation-ordered-exposure}
Let $Y\subseteq\Sigma^c$, and let
$$
Q:=\bigcup_{y\in Y}P(y)
$$
be the set of position-characters appearing in $Y$. Let $X\subseteq\Sigma^c$ be a set of $N$ keys such that $P(x)\nsubseteq Q$ for every $x\in X$. Then, for every dyadic threshold $p=2^{-r}$ satisfying $r\le w$ and $pN\le 1$,
$$
\Pr[\min h(X)\ge p\mid\max h(Y)<p]\ge 1-pN.
$$
\end{lemma}

\begin{proof}
The table values $T[\alpha]$ associated with distinct position-characters $\alpha\in[c]\times\Sigma$ are mutually independent. We may therefore expose them one at a time in any fixed order without changing their joint distribution.

Choose a total order $\prec$ on $[c]\times\Sigma$ in which every position-character in $Q$ precedes every position-character outside $Q$. For each $\alpha\notin Q$, define
$$
G_\alpha:=\{x\in X:\alpha=\max\nolimits_\prec P(x)\}.
$$
Since $P(x)\nsubseteq Q$ for every $x\in X$, each key contains a position-character outside $Q$. Hence its $\prec$-maximum lies outside $Q$, and the nonempty groups $G_\alpha$ form a partition of $X$:
$$
X=\bigsqcup_{\alpha\notin Q}G_\alpha,
\qquad
\sum_{\alpha\notin Q}|G_\alpha|=N.
$$

The event $\{\max h(Y)<p\}$ depends only on the table entries indexed by $Q$. Thus, after conditioning on this event, the table entries indexed outside $Q$ remain mutually independent and uniform. We first expose the entries indexed by $Q$ and then expose the remaining entries according to $\prec$.

For a position-character $\alpha$, let $h(\prec\alpha)$ denote all table values exposed before $T[\alpha]$. Consider a nonempty group $G_\alpha$. For every $x\in G_\alpha$, all position-characters of $x$ other than $\alpha$ precede $\alpha$, and hence
$$
h(x)=T[\alpha]\oplus a_x:=T[\alpha] \oplus \bigoplus_{\beta \in P(x), \beta \prec \alpha} T[\beta],
$$
where $a_x$ is determined by $h(\prec\alpha)$.

Moreover, if $\beta\prec\alpha$ and $x\in G_\beta$, then every position-character of $x$ is at most $\beta$ under $\prec$. Therefore, every previously processed group depends only on table entries contained in $h(\prec\alpha)$. In particular, even after conditioning on all previous groups avoiding $[0,p)$, the value $T[\alpha]$ remains fresh and uniform.

Consequently, for every fixed $x\in G_\alpha$,
$$
\Pr[h(x)<p\mid h(\prec\alpha)]=p.
$$
The hash values within $G_\alpha$ need not be independent, but the union bound gives
$$
\Pr\left[h(G_\alpha)\cap[0,p)=\varnothing\mid h(\prec\alpha)\right]
\ge 1-p|G_\alpha|.
$$
Applying this bound successively to all nonempty groups yields
$$
\Pr[\min h(X)\ge p\mid\max h(Y)<p]
\ge\prod_{\alpha\notin Q}\bigl(1-p|G_\alpha|\bigr)
\ge 1-p\sum_{\alpha\notin Q}|G_\alpha|
=1-pN.
$$
\end{proof}

We are now ready to combine the two lemmas to prove the main theorem.

\begin{proof}[Proof of \Cref{thm:simple-tabulation}]
Fix two distinct characters $0,1\in\Sigma$ and let
$$
Y:=\{(a,b,0,\ldots,0):a,b\in\{0,1\}\}.
$$
Choose any set $X$ of $N$ distinct keys whose first character lies outside $\{0,1\}$. Such a set exists because $(|\Sigma|-2)|\Sigma|^{c-1}\ge N$. The sets $X$ and $Y$ are disjoint.

Let
$$
Q:=\bigcup_{y\in Y}P(y).
$$
For every $x\in X$, the first position-character $(1,x_1)$ does not belong to $Q$. Hence $P(x)\nsubseteq Q$, and \Cref{lem:tabulation-ordered-exposure} applies.

Set
$$
r:=\left\lceil\log_2(2N)\right\rceil
\qquad\text{and}\qquad
p:=2^{-r}.
$$
Since $M=2^w\ge 2N$, we have $r\le w$. Moreover,
$$
\frac{1}{4N}<p\le\frac{1}{2N}.
$$
By \Cref{lem:tabulation-rectangle},
$$
\Pr[\max h(Y)<p]=p^3.
$$
By \Cref{lem:tabulation-ordered-exposure},
$$
\Pr[\min h(X)\ge p\mid\max h(Y)<p]\ge 1-pN\ge\frac12.
$$
The event $\{\max h(Y)<p\}\cap\{\min h(X)\ge p\}$ implies $\{\max h(Y)<\min h(X)\}$. Therefore,
\begin{align*}
    \Pr[\max h(Y)<\min h(X)]
    &\ge
    \Pr[\max h(Y) < p,\,\min h(X) \ge p]\\
    &=
    \Pr[\max h(Y)<p] \cdot \Pr[\min h(X)\ge p\mid\max h(Y)<p]
    \\
    &\ge p^3\cdot\frac12 \ge\frac{1}{128N^3}.
\end{align*}

Under a uniformly random ordering of $X\cup Y$, the probability that the four prescribed elements of $Y$ precede all elements of $X$ is $\binom{N+4}{4}^{-1}$. Hence
$$
\frac{\Pr[\max h(Y)<\min h(X)]}{\binom{N+4}{4}^{-1}}\ge\frac{\binom{N+4}{4}}{128N^3}>\frac{N}{3072}.
$$
The ratio grows linearly with $N$, proving the claim.
\end{proof}

\begin{remark}
The obstruction above inherently uses four prescribed keys. For $|Y|\le3$, the argument of \cite{PT_tabulation} can be adapted to obtain analogous positive guarantees. We do not pursue this direction here.
\end{remark}

\bibliographystyle{alpha} 
\bibliography{rand}

\appendix
\section{Lower Bounds for Seed Length}
\label{sec:support-size-lower-bound}

\begin{proposition}\label{prop:support-size-lower-bound}
Let $1 \le k \le N / 2$ and $0<\delta<1$, and let $\mathcal H$ be a finite $k$-min-wise hash family with multiplicative error $\delta$. Then
$$
|\mathcal H|\ge\max\left\{\binom{N}{k},\min\left\{\frac{1}{\delta},\mathsf{lcm}(1,\ldots,N)\right\}\right\}.
$$
Consequently, the seed length of $\mathcal H$ is
$$
\Omega\left(k\log\frac{N}{k}+\min\left\{\log\frac{1}{\delta},N\right\}\right).
$$
In particular, when $k\le N^{1-c_1}$ and $\delta=N^{-c_2}$, the seed length is $\Omega(k\log N)$.
\end{proposition}
\begin{proof}
Let $D:=|\mathcal H|$. We first prove $D\ge\binom{N}{k}$. Fix any set $Y\subseteq[N]$ of size $k$, and apply the $k$-min-wise guarantee with $X=[N]\setminus Y$. Since $\delta<1$,
$$
\Pr_{h\sim\mathcal H}\left[\max h(Y)<\min h([N]\setminus Y)\right]\ge(1-\delta)\binom{N}{k}^{-1}>0.
$$
Hence some function in $\mathcal H$ realizes $Y$ as the strict bottom-$k$ set of $[N]$. A fixed hash function can realize at most one set as its strict bottom-$k$ set. Since every one of the $\binom{N}{k}$ possible choices of $Y$ must be realized, we obtain $D\ge\binom{N}{k}$.

We next derive a lower bound depending on $\delta$. Fix any integer $m\in[N]$, any set $S\subseteq[N]$ of size $m$, and any element $y\in S$. Since the $k$-min-wise property includes ordinary min-wise hashing and $\mathcal{H}$ is a finite family, there exists an integer $a$ such that
$$
\frac{a}{D}=\Pr_{h\sim\mathcal H}\left[h(y)<\min h(S\setminus\{y\})\right]=(1\pm\delta)\frac{1}{m}.
$$
Therefore,
$$
|ma-D|\le\delta D.
$$
Suppose that $D<1/\delta$. Then $\delta D<1$. Since $ma-D$ is an integer, the preceding inequality forces $ma-D=0$, and hence $m\mid D$. As this holds for every $m\in[N]$, we have $\mathsf{lcm}(1,\ldots,N)\mid D$. Thus either $D\ge1/\delta$, or $D\ge\mathsf{lcm}(1,\ldots,N)$, and consequently
$$
D\ge\min\left\{\frac{1}{\delta},\mathsf{lcm}(1,\ldots,N)\right\}.
$$
Combining the two lower bounds proves the claimed bound on $|\mathcal H|$. Taking logarithms and using $\log\binom{N}{k}=\Omega(k\log(N/k))$ for $k\le N/2$ and $\log\mathsf{lcm}(1,\ldots,N)=\Theta(N)$ (see, e.g., \cite{Sanna2021}) gives
$$
\log D=\Omega\left(k\log\frac{N}{k}+\min\left\{\log\frac{1}{\delta},N\right\}\right).
$$
\end{proof}

\end{document}